\PassOptionsToPackage{dvipsnames}{xcolor}
\documentclass{article}
\usepackage{spconf,amsmath,graphicx}
\usepackage{amsmath}
\usepackage{amsfonts,amsthm,bm}
\usepackage{appendix}
\usepackage{enumitem}
\usepackage{tikz}

\DeclareMathOperator*{\argmin}{arg\,min}

\newcommand{\vu}{\mathbf{u}}
\newcommand{\vb}{\mathbf{b}}
\newcommand{\vy}{\mathbf{y}}

\newcommand{\vr}{\mathbf{r}}
\newcommand{\vx}{\mathbf{x}}
\newcommand{\vwx}{\mathbf{x}}
\newcommand{\vwy}{{\mathbf{y}}}
\newcommand{\vm}{\mathbf{m}}
\newcommand{\vs}{\mathbf{s}}
\newcommand{\bX}{\mathbf{X}}
\newcommand{\bwX}{\bm{\mathcal{X}}}
\newcommand{\bwY}{\bm{\mathcal{Y}}}
\newcommand{\bY}{\mathbf{Y}}

\newcommand{\vW}{\mathbf{W}}

\newcommand{\vM}{\mathbf{M}}
\newcommand{\vY}{\mathbf{Y}}

\newcommand{\vX}{\mathbf{X}}
\newcommand{\vA}{\mathbf{A}}

\newcommand{\vS}{\mathbf{S}}

\newcommand{\vG}{\mathbf{G}}
\newcommand{\vP}{\mathbf{P}}
\newcommand{\vD}{\mathbf{D}}
\newcommand{\vI}{\mathbf{I}}

\newcommand{\bGam}{\bm{\Gamma}}
\newcommand{\bSig}{\bm{\Sigma}}
\newcommand{\bmu}{\bm{\mu}}
\newcommand{\bTheta}{\bm{\Theta}}
\newcommand{\bPi}{\bm{\Pi}}
\newcommand{\bPhi}{\bm{\Phi}}
\newcommand{\bLambd}{\bm{\Lambda}}
\newtheorem{theorem}{Theorem}
\def\CLW{0.5mm}  
\def\FSLW{0.4mm}  
\def\RSLW{0.5mm}  

\newif\ifsparse

\title{Blind Bounded Source Separation Using Neural Networks with Local Learning Rules}
\twoauthors
  { Alper T. Erdogan}
	{Electrical-Electronics Engineering \\
	Koc University\\
	Istanbul, Turkey}
  {Cengiz Pehlevan\sthanks{thanks Intel Corporation for funding of this work.}}{School of Engineering and Applied Sciences  \\ Harvard University\\ Cambridge, MA USA}
\begin{document}
%
\maketitle
\begin{abstract}
An important problem encountered by both natural and engineered signal processing systems is blind source separation. In many instances of the problem, the sources are bounded by their nature and known to be so, even though the particular bound may not be known. To separate such bounded sources from their mixtures, we propose a new optimization problem, Bounded Similarity Matching (BSM). 
A principled derivation of an adaptive BSM algorithm leads to a recurrent neural network with a clipping nonlinearity. The network adapts by local learning rules, satisfying an important constraint for both
biological plausibility and implementability in neuromorphic
hardware. 
\end{abstract}
\begin{keywords}
Similarity Matching, Recurrent Neural Networks, Local Update Rule, Blind Source Separation, Bounded Component Analysis.
\end{keywords}
\section{Introduction}
\label{sec:intro}

Blind source separation is a fundamental problem for natural and engineered signal processing systems. In this paper, we show how it can be solved by a class of neural networks important for both neuroscience and machine learning, i.e. those with local learning rules, where the strength of a synapse is updated based on the activations of only the pre- and postsynaptic neurons. Locality of learning rules is a natural constraint for  biological neural networks \cite{dayan2005theoretical, kandel2000principles}, and enables large scale implementations of neural network algorithms with low power in recent neuromorphic integrated circuits \cite{lin2018programming}.

Similarity Matching has  been introduced as a gradient-based optimization framework   for  principled derivation of neural networks with local learning rules \cite{pehlevan2015normative, pehlevanSPMtoappear}. This framework can be used to provide solutions for clustering \cite{pehlevan2014hebbian,bahroun2017neural}, sparse feature extraction \cite{ pehlevan2014hebbian} and manifold learning \cite{sengupta2018manifold}. It was also applied to a nonnegative blind source separation problem \cite{pehlevan2017blind} which leads to a recurrent neural network with ReLU activation functions. 

Similar to nonnegativity, spatial boundedness is a property that can be exploited to separate sources from their linear mixtures. Bounded Component Analysis (BCA) has been introduced as a framework exploiting this property to separate both dependent and independent sources from their linear mixtures \cite{cruces2010bounded,erdogan2013class}.  It was successfully applied to separation of dependent natural images, digital communication signals, as well as sparse bounded sources \cite{erdogan2013class, babatas2018algorithmic, babatas2018sparse}.

This article has two main contributions. First, we formulate a new optimization problem, bounded similarity matching (BSM), for blind bounded source separation (BBSS). By using diagonally weighted inner products and bounded outputs, we show that the BBSS problem can be formulated as a minimization of the inner product weights under a weighted and bounded similarity matching constraint. Second, by an online optimization of this problem, we derive a  biologically plausible recurrent neural network with clipping activation functions, whose parameters are updated by a local learning rule. The update rules of synaptic weights parallel the plasticity mechanisms observed in biological neurons.

\section{BBSS Setting}
\label{sec:BSSseting}
We consider the following BBSS scenario:
\begin{itemize}[leftmargin=*]
	\item There are $d$ sources, represented by the vector sequence $\{\vs_t \in \mathbb{R}^d, t\in \mathbb{Z}^+\}$.
	Sources are uncorrelated, 
	\begin{align}
	\bSig_{\vs} &=E((\vs_t -\bmu_\vs)(\vs_t-\bmu_\vs)^{\top})\nonumber \\
	&=\text{diag}(	\sigma_{s^{(1)}}^2,	\sigma_{s^{(2)}}^2, \ldots, 	\sigma_{s^{(d)}}^2),
	\end{align}
	where $\bmu_{\vs}=E(\vs)$,  and bounded around their mean
	\begin{eqnarray}
	\bmu_{\vs}-\vr/2 \le \vs_t \le \bmu_\vs+\vr/2, \hspace{0.3in} \forall t \in \mathbb{Z}^+,
	\end{eqnarray}
	where $\vr=\left[r^{(1)}, r^{(2)}, \ldots,r^{(d)}\right]^{\top}$, $r^{(i)}$ the range of source $s^{(i)}$. 
	The bounds are unknown to the BBSS algorithm.
	
	
	\item Sources are mixed through a linear memoryless system,
$\vm_t=\vA \vs_t$, $ \forall t \in \mathbb{Z}^+$,
	where $\vA \in \mathbb{R}^{q \times d}$ is the full rank mixing matrix, and $\vm_t \in \mathbb{R}^q$ are mixtures. We assume an (over)determined mixing, i.e., $q\ge d$. We define $\bmu_{\vm} = E(\vm_t)$.
	
	\item The mixtures are pre-processed with whitening and mean-removal, which can be implemented through various adaptive mechanisms including biologically plausible neural networks, see e.g. \cite{pehlevan2015normative}. The pre-processed mixtures can be written as
	\begin{eqnarray}
	\hspace*{-0.1in}\vx_t=\vW_{pre}(\vm_t-\bmu_{\vm})= \bTheta \bSig_{\vs}^{-1/2}(\vs_t-\bmu_{\vs})= \bTheta \bar{\vs}_t, 
	\end{eqnarray}
	where $\bar{\vs}_t$ is the scaled and mean-removed sources that satisfy
	$ E(\bar{\vs})=\mathbf{0}$, and 
	$E(\bar{\vs}\bar{\vs}^T)=\mathbf{I}$.
	Furthermore, $\bar{\vs}_t$ is symmetrically bounded around zero satisfying
	$-\vb \le \bar{\vs}_t\le \vb$,
	where $\vb=\bSig_{\vs}^{-1/2}\frac{\vr}{2}$. Here, note that since $\sigma_{s^{(i)}}\le\frac{r^{(i)}}{2}$, bounds $b^{(i)}$'s are greater than or equal to  $1$.  $\bTheta$ is a real orthogonal matrix satisfying
	$\bTheta^T\bTheta=\vI$ and $\vx_t$, like $\bar{\vs}_t$, is a white signal with unit variance.
	
\end{itemize}


\section{BSM: A New Optimization Formulation of BBSS}

Our solution to BBSS takes inspiration from the Nonnegative Similarity Matching (NSM) method given in \cite{pehlevan2017blind}, which proposed the following optimization problem to recover nonnegative sources from their whitened mixtures:
\begin{eqnarray}
\vY^*=\argmin_{ \vY \ge \mathbf{0}} \|\vX^{\top}\vX -\vY^{\top}\vY \|_F^2,
\end{eqnarray}
where $\bX :=\left[\vwx_1, \vwx_2,\ldots,  \vwx_{t-1}, \vwx_t \right]$ is the data matrix and $\bY :=\left[\vwy_1,\vwy_2, \ldots,  \vwy_{t-1}, \vwy_t \right]$ is the output matrix. NSM amounts to finding nonnegative $\vy_n$ vectors which have the same decorrelated and variance-equalized structure as whitened-mixtures. 

If we were to adopt a similar approach for bounded but potentially mixed-sign sources, we could instead propose
\begin{eqnarray}
\vY^*=\argmin_{ -\vb \le \vY \le \vb} \|\vX^{\top}\vX -\vY^{\top}\vY \|_F^2.
\end{eqnarray}
However, this form requires knowledge of the source bounds, contrary to our setting. Instead, we aim to extract whitened sources up to scale factors, which are defined as
$\tilde{\vs}_t=\bPi_\vb^{-1}\bar{\vs}_t$,
where $\bPi= \text{diag}(b^{(1)}, b^{(2)}, \ldots, b^{(d)})$, and the resulting scaled sources satisfy $-\mathbf{1} \leq \tilde{\vs} \leq \mathbf{1}$. Therefore, we restrict outputs to a fixed range, such as the unity $\ell_\infty$-norm ball defined as
$\mathcal{B}_\infty=\{\vy : \|\vy\|_\infty\le 1 \}$.

In light of these arguments, we pose the following optimization problem, Bounded Similarity Matching (BSM), for bounded source separation:  
\begin{eqnarray}
\underset{\vY, D_{11}, \ldots D_{dd}} {\text{minimize}}&& \sum_{i=1}^d D_{ii}^2  \hspace{1.5in}  \text{\it BSM} \nonumber\\
\text{subject to} &&  \vX^{\top}\vX -\vY^{\top}\vD\vY=0  \nonumber \\
&&-\mathbf{1} \le \vY \le \mathbf{1} \nonumber \\
&& \vD=\text{diag}(D_{11}, \ldots, D_{dd}) \nonumber \\
&& D_{11}, D_{22}, \ldots D_{dd}>0 \nonumber
\end{eqnarray}
Note that the similarity metric of BSM is different than that of NSM. In order to confine outputs to a fixed dynamic range, we propose  to match their weighted inner products to the inner products of inputs. This enables imposing boundedness of sources without knowing exact bounds. 
Estimation of individual bounds is achieved by optimizing over the inner product weights $\vD$.

%

The following theorem asserts that the global minima  of BSM are perfect separators:
\begin{theorem} 
	Given the BBSS setup described in Section \ref{sec:BSSseting}, if  the vector sequence $\{\tilde{s}_t\}$ contains all corner points of $\mathcal{B}_\infty$, then the global optima for  BSM
	satisfy
	\begin{eqnarray}
	\vY&=&\bLambd \vP \vD^{-1/2}\bar{\vS}=\bLambd \vP\tilde{\vS}, \\
	D_{ii}&=&b_{J_i}^2, \hspace{0.3in} i=1,\ldots d,
	\end{eqnarray}
	where $\bLambd$ is a diagonal matrix with $\pm 1$'s on the diagonal (representing sign ambiguity),  $\vP$ is a permutation matrix, $J$ is a permutation of $\{1,2, \ldots, d\}$.
\end{theorem}
\begin{proof}
	The proof is  in Appendix \ref{ap:proofofthm}.
\end{proof}

\section{Adaptive BSM via a Neural Network with local learning rules}
In this section, we derive an  adaptive implementation  of BSM and  show that it corresponds to a biologically plausible neural network with local update rules. For this purpose, our first task is to introduce an online version of the  BSM optimization problem introduced in the previous section. 

In the online setting, we consider exponential weighting of the signals  for dynamical adjustment and define  the weighted  input data  snapshot matrix  by time $t$ as,
	\begin{eqnarray}
	\bwX_t&=&\left[\begin{array}{ccccc}  \gamma^{t-1} \vwx_1 & \gamma^{t-2} \vwx_2 & \ldots & \gamma \vwx_{t-1} & \vwx_t \end{array}\right]\nonumber \\
	&=& \bX_t\bGam_t,
	\end{eqnarray}
	where $\bGam_t=\text{diag}( \gamma^{t-1}, \gamma^{t-2}, \ldots, 1)$. Similarly, we define the weighted output snapshot matrix by time $t$ as,
	\begin{eqnarray}
	\bwY_t&=&\left[\begin{array}{ccccc}  \gamma^{t-1} \vwy_1 & \gamma^{t-2} \vwy_2 & \ldots & \gamma \vwy_{t-1} & \vwy_t \end{array}\right] \nonumber\\
	&=& \bY_t\bGam_t
	\end{eqnarray}	
	We  define $\kappa_t= \sum_{k=0}^{t-1}\gamma^{2k}=\frac{1-\gamma^{2t}}{1-\gamma^{2}}$ as a measure of  effective time window length for sample correlation calculations, given the exponential weights. Assuming  sufficiently large $t$, we take $\kappa=\frac{1}{1-\gamma^2} $.

In order to define on online BSM cost function, we replace the hard equality constraint on BSM with a square error minimization, and introduce the weighted similarity matching (WSM) cost function as,
\begin{eqnarray}
J_{WSM}(\bwY_t,\vD_t)=\frac{1}{\kappa^2}\|\bwX_t^T\bwX_t-\bwY_t^T\vD_t\bwY_t\|_F^2.
\end{eqnarray}
With this definition, we consider a relaxation of {\it BSM} (rBSM):
\begin{eqnarray}
J_{rBSM}(\bwY_t,\vD_t)=J_{WSM}(\bwY_t,\vD_t)+2\alpha_D \sum_{i=1}^d D_{t,ii}^2,
\end{eqnarray}
to be minimized under the constraint set $-1\le \bwY_t\le 1$.

%
%

Following a treatment similar to \cite{pehlevan2017blind}, we derive an online cost from $J_{rBSM}$ by considering its optimization with respect to the data already received and only with respect to the current output, $\vy_t$, and weights, $\vD_t$. This reduces to minimizing,
\begin{align}
	h(\vy_t,\vD_t)= &\kappa\text{Tr}(\vM_t\vD_t\vM_t\vD_t)-2\kappa\text{Tr}(\vW_t^T\vD_t\vW_t) \nonumber \\
	&\hspace*{-0.65in} +2\vy_t^T\vD\vM_t\vD_t\vy_t-4\vy_t^T\vD_t\vW_t\vx_t+2\alpha_D \sum_{i=1}^d D_{t,ii}^2,
\end{align}
where
\begin{eqnarray}\label{weights}
\vW_t=\frac{1}{\kappa} \bwY_{t-1}\bwX_{t-1}^T=\frac{1}{\kappa}\sum_{k=1}^{t-1}(\gamma^2)^{t-1-k}\vy_k\vx_{k}^T, \nonumber \\
\vM_t=\frac{1}{\kappa} \bwY_{t-1}\bwY_{t-1}^T=\frac{1}{\kappa}\sum_{k=1}^{t-1}(\gamma^2)^{t-1-k}\vy_k\vy_{k}^T,  
\end{eqnarray}
subject to $-1\leq \vy_t \leq 1$.

In order to minimize this cost function, we write its gradient with respect to $\vy_t$ as
\begin{eqnarray}
\frac{1}{4}\nabla_{\vy_t}h(\vy_t,\vD_t)=\vD_t\vM_t\vD\vy_t-\vD_t\vW_t\vx_t,
\end{eqnarray}
and the gradient with respect to $D_{t,ii}$ as
\begin{align}\label{dgrad}
\frac{1}{4}\nabla_{D_{t,ii}}h(\vy_t,\vD_t)=&\frac{\kappa}{2}(\|{\vM_{t+1}}_{i,:}\|_\vD^2-\|{\vW_{t+1}}_{i,:}\|_2^2)\nonumber \\&+\alpha_D D_{t,ii}
\end{align}
Note that, as $\vD_t>0$,  
\begin{eqnarray}
-\vD_t^{-1} \frac{1}{4}\nabla_{\vy_t}h(\vy_t,\vD_t)=\vW_t\vx_t-\vM_t\vD_t\vy_t \label{descentdir}
\end{eqnarray}
provides a local descent direction with respect to $\vy_t$. We decompose $\vM_t=\bar{\vM}_t+\Upsilon_t$ where $\Upsilon_t$ is the diagonal, nonnegative matrix containing diagonal components of $\vM_t$. 
As a result the the last term  in the descent direction can be rewritten as %
$\bar{\vM}_t\vD\vy_t+\Upsilon_t\vD_t\vy_t$.
We further define $\vu=\Upsilon_t\vD_t\vy_t$. Since $\vu$ and $\vy_t$ are monotonically related, (\ref{descentdir}) is a descent direction with respect to $\vu$. Using a gradient search with small step size,  we obtain a  neural network form for the gradient search algorithm
\begin{align}\label{dynamics}
\frac{d\vu(\tau)}{d\tau}&=-\vu(\tau)+\vW_t \vx_t-\bar{\vM}_t\vD_t\vy_t(\tau), \nonumber \\
\vy_{t,i}(\tau)&=c\left(\frac{\vu_i(\tau)}{{\Upsilon_{t,ii}}{\vD_{t,ii}}} \right),
\end{align}
where $c$ is the clipping function corresponding to the projection of output components  to the constraint set $[-1,1]$, which can be written as 
\begin{eqnarray}\label{clipfun}
c(z)=\left\{ \begin{array}{cc} z &  -1\le z \le 1, \\
{\rm sign}(z) & \text{otherwise.} \end{array}  \right.
\end{eqnarray}
The corresponding recurrent neural network is shown in Fig. \ref{fig:RNNBSM}.
It is interesting to observe that the  inverse of the inner product weights act as activation function gains.

After the neural network dynamics of \eqref{dynamics} converges, synaptic and inner product weights are updated for the next input by
\begin{align}
\vW_{t+1}&=\gamma^2 \vW_t + (1-\gamma^2)\vy_t\vx_t^T \nonumber\\
\vM_{t+1}&=\gamma^2 \vM_t+(1-\gamma^2)\vy_t\vy_t^T \nonumber\\
{\vD_{t+1,ii}}&=(1-\beta) {\vD_{t,ii}} +\eta (\|{\vW_{t+1}}_{i,:}\|_2^2-\|{\vM_{t+1}}_{i,:}\|_{\vD_t^2}), \nonumber
\end{align}
where $\eta$ is the step size and $\beta=2\eta\alpha_D$.
The synaptic weight updates follow from their definitions in \eqref{weights}. In neuroscience, $\vW$ updates are called Hebbian, and $\vM$ updates are called anti-Hebbian (because of the minus sign in \eqref{dynamics}) \cite{foldiak1990forming,dayan2005theoretical}.  The gain updates turn out to be in the form of a leaky integrator due to the last term in \eqref{dgrad}.
It is interesting to observe that these gain updates are negative or positive depending on the balance between the  norms of feedforward and recurrent weights. These weights are the reflectors of the recent excitation and inhibition statistics respectively (as they are corresponding correlation matrices). Relative  increase in excitation(inhibition) would cause increase(decrease) in weights, and therefore,   decrease(increase) in the corresponding homeostatic gains. This resembles the experimentally observed homeostatic gain adjustment behavior in  biological neurons \cite{turrigiano1999homeostatic}.

\begin{figure}
	\begin{center}
		\begin{tikzpicture}[thick, scale=0.38]
		\clip (0,0) rectangle + (14,13);
		\coordinate (A) at (10,11);
		
		\coordinate (S11)   at  (10.88,10.1) ;
		\coordinate (S12) at (11,8) ;
		
		\coordinate (S21) at (11.02,10.9);
		\coordinate (S22) at (10.9,8.95);
		
		\coordinate (SM1) at (11.02,3);
		\coordinate (S1M) at (10.92,11.67);
		
		\coordinate (SM2) at (10.88,3.70);
		\coordinate (S2M) at (10.95,7.72);

		\node [left] at (2.5,10.5) {\LARGE $x_1$};
		\filldraw[thick] (2.5,10.5) circle (0.05cm);
		
		\node [left] at (2.5,7.5) {\LARGE $x_2$};
		\filldraw[thick] (2.5,7.5) circle (0.05cm);

		\node [left] at (2.5,3.5) {\LARGE $x_n$};
		\filldraw[thick] (2.5,3.5) circle (0.05cm);
		\draw[thick, dotted] (2,6)--(2,5);

		\draw[thick] (A) circle (1.0cm);
		\draw[<->, line width=0.4mm] (9.1,11) -- (10.9,11);
		\draw[<->, line width=0.4mm] (10,10.2) -- (10,11.8);
		\draw[-, line width={\CLW},blue] (9.2,10.6) -- (9.65,10.6)-- (10.4,11.4)--(10.8,11.4);
		
		\draw[thick] (10,8) circle (1.0cm);
		\draw[<->, line width=0.4mm] (9.1,8) -- (10.9,8);
		\draw[<->, line width=0.4mm] (10,7.2) -- (10,8.8);
		\draw[-, line width={\CLW},blue] (9.2,7.6) -- (9.65,7.6)-- (10.4,8.4)--(10.8,8.4);
		\draw[thick] (10,3) circle (1.0cm);
		\draw[<->, line width=0.4mm] (9.1,3) -- (10.9,3);
		\draw[<->, line width=0.4mm] (10,2.2) -- (10,3.8);
		\draw[-, line width={\CLW},blue] (9.2,2.6) -- (9.65,2.6)-- (10.4,3.4)--(10.8,3.4);
		
		\draw[thick, dotted] (10,6.3)--(10,4.7);
		
		\draw [ line width={\RSLW}, WildStrawberry]   (S11) to[out=-30,in=3] (S12);
		\draw [ line width={\RSLW},WildStrawberry]   (S21) to[out=-5,in=50] (S22);
		\draw [ line width={\RSLW}, WildStrawberry]   (SM1) to[out=2,in=50] (S1M);
		\draw [ line width={\RSLW}, WildStrawberry]   (S2M) to[out=-30,in=50] (SM2);
		\filldraw[thick,WildStrawberry] (10.94,3.74) circle (0.15cm);
		\filldraw[thick,WildStrawberry] (10.82,8.85) circle (0.15cm);
		\filldraw[thick,WildStrawberry] (10.80,10.14) circle (0.15cm);
		\filldraw[thick,WildStrawberry] (10.94,11.72) circle (0.15cm);
		
		\draw[-, line width={\FSLW}] (2.5,10.5) -- (8.69,10.93);
		\draw[-, line width={\FSLW}] (2.5,10.5) -- (8.75,8.51);
		\draw[-, line width={\FSLW}] (2.5,10.5) -- (9.09,3.95);
		
		\draw[-, line width={\FSLW}] (2.5,7.5) -- (8.8,10.45);
		\draw[-, line width={\FSLW}] (2.5,7.5) -- (8.69,7.95);
		\draw[-, line width={\FSLW}] (2.5,7.5) -- (8.80,3.53);
		
		\draw[-, line width={\FSLW}] (2.52,3.5) -- (9.13,9.99);
		\draw[-, line width={\FSLW}] (2.52,3.5) -- (8.9,7.27);
		\draw[-, line width={\FSLW}] (2.52,3.5) -- (8.68,3.03);
		
		\draw[thick] (8.83,10.95) circle (0.15cm);
		\draw[thick] (8.92,8.45) circle (0.15cm);
		\draw[thick] (9.2,3.85) circle (0.15cm);
		
		\draw[thick] (8.92,10.52) circle (0.15cm);
		\draw[thick] (8.83,7.95) circle (0.15cm);
		\draw[thick] (8.92,3.45) circle (0.15cm);
		
		\draw[thick] (9.25,10.1) circle (0.15cm);
		\draw[thick] (9.02,7.35) circle (0.15cm);
		\draw[thick] (8.83,3.) circle (0.15cm);

		\draw[thick] (3,1.45) circle (0.15cm);
		\node [right] at (3.3,1.45) {\large $\text{Hebbian}$};
		\filldraw[thick,WildStrawberry] (3,0.45) circle (0.15cm);
		\node [right] at (3.3,0.45) {\large $\text{Anti-Hebbian}$};
		
	    \filldraw[draw=black,color=white,opacity=1.0] (8,8) rectangle (4,6);
		\node [left] at (7.5,7) {\huge $\mathbf{W}$};
	    \filldraw[draw=black,color=white,opacity=1.0] (15,8.3) rectangle (11.5,5.8);
		\node [right,WildStrawberry] at (10.1,7) {\huge $\mathbf{-\bar{M}}$};
		\end{tikzpicture}
	\end{center}
	\caption{Recurrent Neural Network for BSM.}
	\label{fig:RNNBSM}
\end{figure}
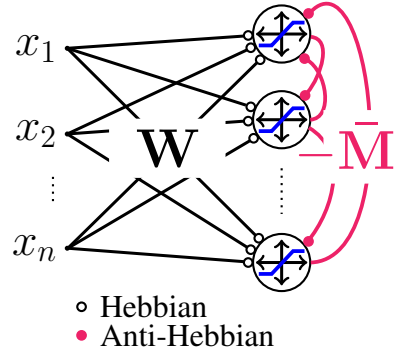

\section{Numerical Examples}
\subsection{Source Separation Example}
As an illustration of the bounded source separation capability of the recurrent BSM network, we performed the following numerical experiment:
	We generated $10$ uniform sources from $\mathcal{U}[0,B]$ where the maximum value $B$ is chosen randomly for each choice from $[2,7]$ interval randomly (uniformly chosen). The sources were mixed by a random real orthogonal matrix, representing the whitened mixtures case. 
	We selected the synapse update parameter  $1-\gamma^2=4\times10^{-3}$, the homeostatic gain update parameter as $\eta= 10^{-3}$ and the leaky integral parameter for homeostatic gain parameter as $\beta=10^{-6}$.

Fig. \ref{fig:example10sources}.(a) shows the Signal-to-Interference Ratio energy (SIR), which is the total energy of source signals at the outputs to the energy of interference from other sources. The SIR increases with iterations after an initial transient and the clipping input's peak converges towards the clipping level.  Homeostatic gains settle, as can be seen in both weight convergence curves in Fig. \ref{fig:example10sources}.(c) and the excitation/inhibition curves in Fig. \ref{fig:example10sources}.(d).

\begin{figure}
	\centering
	\includegraphics[width=0.99\linewidth]{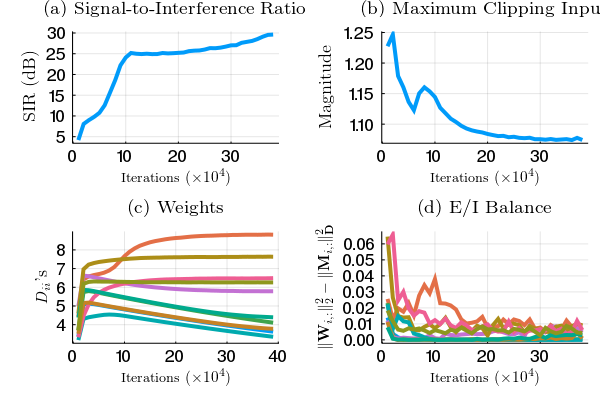}
	\caption{ $10$ Uniform Source Example:   (a) Signal-to-Interference Ratio (SIR) for outputs, (b) the maximum magnitude input (among all neurons) to the clipping function, (c) Weights($D_{ii}$'s), (d) E/I balance for each neuron as measured based on excitation and inhibition synaptic strength norms.}
	\label{fig:example10sources}
\end{figure}

\subsection{Image Separation}
We consider the problem of separating 3 images selected from the database in \cite{martin2002database} from their $3$  random unitary mixtures. As illustrated in Fig. \ref{fig:exampleimagesep}, we obtain outputs that are very similar to original nearly uncorrelated sources, with an SIR level around $30$dB.
\begin{figure}
	\centering
	\includegraphics[width=0.99\linewidth]{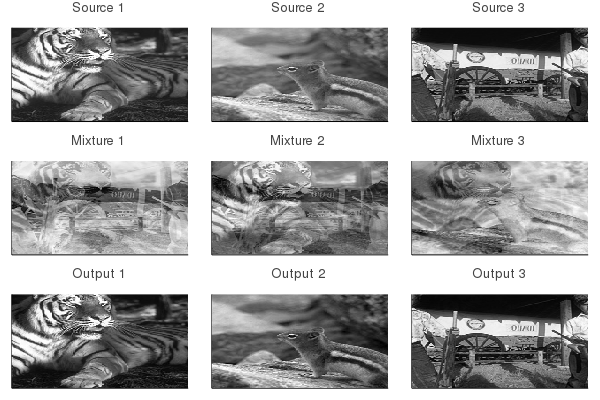}
	\caption{Image Separation Example.}
	\label{fig:exampleimagesep}
\end{figure}

\begin{appendices}
	
	\section{Proof of Theorem I}
	\label{ap:proofofthm}
	We start by observing that, since $\vx_t=\bTheta\bar{\vs}_t$ and $\bTheta$ is real orthogonal, the equality constraint of {\it BSM} is equivalent to
	$\bar{\vS}^T\bar{\vS}-\vY^T\vD\vY=0$,
	where $\bar{\vS}$ contains the sequence $\bar{\vs}_t$ in its columns. Referring to Theorem I in \cite{pehlevan2015normative}, this equality constraint implies
	\begin{eqnarray}
	\vD^{1/2}\vy_t&=&\vG \bar{\vs}_t= \underbrace{\vG\bPi_\vb}_{\bPhi}\tilde{\vs}_t.
	\end{eqnarray}
	The $i^{th}$ component of left hand side is maximized for index $m$ where $\vs_m=\text{sign}(\bPhi_{i,:})^T$, i.e.,
	\begin{eqnarray}
	D_{ii}\max_t(y^{(i)}_t)&=&\bPhi_{i,:}\text{sign}(\bPhi_{i,:})^T =  \|\bPhi_{i,:}\|_1
	\end{eqnarray}
	As assumed in Theorem 1, since the vector sequence $\{\tilde{s}_t\}$ contains the corners of $\mathcal{B}_\infty$, which are all possible sign patterns, we can write 
	\begin{eqnarray}
	D_{ii}=\frac{\|\bPhi_{i,:}\|_1}{\max_t(y^{(i)}_t)}, \hspace{0.5in} i=1, \ldots, d. 
	\end{eqnarray}
	In terms of this expression, we can rewrite the cost function for {\it BSM}, and obtain a lower bound:
	\begin{eqnarray}
	\sum_{i=1}^d D_{ii}^2&=&\sum_{i=1}^d \frac{\|\bPhi_{i,:}\|_1^2}{\max_t(y^{(i)}_t)^2}\\
	&\ge& \sum_{i=1}^d \|\bPhi_{i,:}\|_2^2   \label{eq:boundineq}\\
	&=&   \sum_{i=1}^d (b^{(i)})^2\|\vG(:,i)\|_2^2 = \sum _{i=1}^d(b^{(i)})^2,
	\end{eqnarray}
	where  the inequality is due to $\max_t(y^{(i)}_t)\le 1$ and  $\ell_p$-norm inequality, and the last equality is due to the fact that $\vG$ is real orthogonal. The lower bound is achieved if and only if the  inequality in (\ref{eq:boundineq}) is equality. This condition is achieved if all rows of $\bPhi$, and therefore, $\vG$ have only one nonzero element. Therefore, for the optimal solution,
	$\vG$ has the form
	$\vG=\bLambd\vP$,
	where $\bLambd$ and $\vP$ are as defined in the theorem. This implies that optimal $D_{ii}$s are equal to  permuted versions of $b_i$s. As a result, the global optimal solutions to {\it BSM}  can be characterized by the statement of the theorem.
	

\end{appendices}
\newpage
\bibliographystyle{IEEEbib}

\end{document}